\documentclass[runningheads,11points]{llncs}

\usepackage[russian,greek,english]{babel}
\usepackage[utf8x]{inputenc}
\usepackage{amsmath,ifthen,color,eurosym}

\usepackage{wrapfig,hyperref,cite}
\usepackage{latexsym,amsmath,amssymb,url}
\usepackage{fancyhdr}
\usepackage{graphicx}
\usepackage{stmaryrd,wasysym}
\usepackage{colortbl,datetime}
\usepackage{boxedminipage}

\usepackage[lined,boxed,commentsnumbered]{algorithm2e}

\newcommand{\remove}[1]{}

\spnewtheorem{observation}{Observation}{\bfseries}{\itshape}
\spnewtheorem{claimN}{Claim}{\bfseries}{\itshape}

\spnewtheorem{remarkLI}{Remark}{\bfseries}{\itshape}

\newcommand{\PartIntSup}[1]{\left\lceil #1\right\rceil}
\newcommand{\PartIntInf}[1]{\left\lfloor #1\right\rfloor}

\renewenvironment{proof}{\noindent \textsc{Proof:}}{\hfill$\square$\medskip}
\newenvironment{proofof}{\noindent \textsc{Proof:}}{\hfill$\Diamond$\medskip}

\title{On the complexity of computing the\\$k$-restricted edge-connectivity of a graph\thanks{An extended abstract of this article appeared in the \emph{Proceedings of the 41st International Workshop on Graph-Theoretic Concepts in Computer Science} (\textbf{WG}), pages 219-233, volume 9224 of LNCS, Garching, Germany, June \textbf{2015}.}
}
\author{Luis Pedro Montejano~\inst{1} and Ignasi Sau~\inst{2}}

\authorrunning{L. P. Montejano and I. Sau}
\titlerunning{Computing the $k$-restricted edge-connectivity of a graph}

\institute{D\'epartement de Math\'ematiques, Universit\'e de Montpellier, Montpellier, France.\\
\email{lpmontejano@gmail.com}
\and AlGCo project team, CNRS, LIRMM, Montpellier, France.\\
\email{ignasi.sau@lirmm.fr}}

\begin{document}

\maketitle
\setcounter{footnote}{0}

\vspace{-.25cm}

\begin{abstract}
The \emph{$k$-restricted edge-connectivity} of a graph $G$, denoted by $\lambda_k(G)$, is defined as the minimum size of an edge set whose removal leaves exactly two connected components each containing at least $k$ vertices. This graph invariant, which can be seen as a generalization of a minimum edge-cut, has been extensively studied from a combinatorial point of view. However, very little is known about the complexity of computing $\lambda_k(G)$. Very recently, in the parameterized complexity community the notion of \emph{good edge separation} of a graph has been defined, which happens to be essentially the same as the $k$-restricted edge-connectivity. Motivated by the relevance of this invariant from both combinatorial and algorithmic points of view, in this article we initiate a systematic study of its computational complexity, with special emphasis on its parameterized complexity for several choices of the parameters. We provide a number of {\sf NP}-hardness and {\sf W}[1]-hardness results, as well as {\sf FPT}-algorithms.

 \vspace{0.25cm}


\end{abstract}


\section{Introduction}
\label{sec:intro}

\textbf{Motivation}. The $k$-restricted edge-connectivity is a graph invariant that has been widely studied in the literature from a combinatorial point of view~\cite{BCFF97,BUV02,FF94,HMM12,ZY05,YL10}.
Since the classical edge-connectivity may not suffice to
measure accurately how connected a graph is after deleting some
 edges, Esfahanian and Hakimi \cite{EsHa88} proposed in 1988 the
notion of {\it restricted edge-connectivity}. Given a graph $G$, a non-empty set $S \subseteq E(G)$ is an \emph{edge-cut}  if $G-S$ has at least two connected components. An edge-cut $S$ is called a
{\it restricted edge-cut} if there are no isolated vertices in $G-S$. The {\it
restricted edge-connectivity}
$\lambda'(G)$ is the minimum cardinality over all {\it restricted
edge-cuts} $S$.

Inspired by the above definition, F\`abrega and Fiol
\cite{FF94} proposed in 1994 the notion of \emph{$k$-restricted
edge-connectivity}, where $k$ is a positive integer, generalizing
this notion. 
An edge-cut $S$ is called a
{\it $k$-restricted edge-cut} if every component of $G-S$ has at
least $k$ vertices. Assuming that $G$ has $k$-restricted edge-cuts,
the {\it $k$-restricted edge-connectivity} of $G$, denoted by
$\lambda_{k} (G)$, is  defined as the minimum cardinality over all
$k$-restricted edge-cuts of $G$, i.e.,
$$\lambda_{k}(G)=\min\{|S|\colon S\subseteq E(G) \hbox{ is a $k$-restricted edge-cut}\}.$$

Note that for any graph $G$, $\lambda_1(G)$ is the size of a minimum edge-cut, and $\lambda_2(G) = \lambda'(G)$. A connected graph $G$ is called {\it $\lambda_{k}$-connected} if
$\lambda_{k}(G)$ exists. Let $[X,Y]$ denote the set of edges between two disjoint vertex sets $X,Y\subseteq V(G)$, and let $\overline{X}$ denote the complement $\overline{X}=V(G)\setminus X$ of vertex set $X$. It is clear that for any
$k$-restricted cut $[X, \overline{X}]$ of size $\lambda_k(G)$, the graph $G-[X,
\overline{X}]$ has exactly two connected components.


Very recently, Chitnis \emph{et al}.~\cite{CCH+12} defined the notion of \emph{good edge separation} for algorithmic purposes. For two positive integers $k$ and $\ell$, a partition $(X,\overline{X})$ of the vertex set of a connected graph $G$ is called a \emph{$(k,\ell)$-good edge separation} if $|X|,|\overline{X}| > k$, $|[X,\overline{X}]| \leq \ell$, and both $G[X]$ and $G[\overline{X}]$ are connected. That is, it holds that $\lambda_k(G) \leq \ell$ if and only if $G$ admits a $(k-1,\ell)$-good edge separation. Thus both notions, which have been defined independently and for which there existed no connection so far, are essentially the same.

Good edge separations turned out to be very useful for designing parameterized algorithms for cut problems~\cite{CCH+12}, by using a technique known as \emph{recursive understanding}, which basically consists in breaking up the input graph into highly connected pieces in which the considered problem can be {\sl efficiently} solved. It should be mentioned that Kawarabayashi and Thorup~\cite{KaTh11} had defined before a very similar notion for {\sl vertex-cuts} and introduced the idea of recursive understanding.  This technique has also been subsequently used in~\cite{CLP+14,KOP+15,RRS16}.


Very little is known about the complexity of computing the $k$-restricted edge-connectivity of a graph, in spite of its extensive study in combinatorics. In this article we initiate a systematic analysis on this topic, with special emphasis on the parameterized complexity of the problem. In a nutshell, the main idea is to identify relevant parameters of the input of some problem, and study how the running time of an algorithm solving the problem depends on the chosen parameters. See~\cite{Cygan15book,DF99,FG06,Nie06} for introductory textbooks to this area.


\textbf{Our results}. We consider the following two problems concerning the $k$-restricted edge-connectivity of a graph.

\vspace{.4cm}

\begin{boxedminipage}{.92\textwidth}
\textsc{Existential Restricted Edge-connectivity (EREC)}\vspace{.1cm}

\begin{tabular}{ r l }
\textbf{~~~~Instance:} & A graph $G=(V,E)$  and a positive integer $k$. \\
\textbf{Question:} & Is $G$ $\lambda_k$-connected ?\\
\end{tabular}
\end{boxedminipage}

\vspace{.4cm}
\begin{boxedminipage}{.92\textwidth}
\textsc{Restricted Edge-connectivity (REC)}\vspace{.1cm}

\begin{tabular}{ r l }
\textbf{~~~~Instance:} & A connected graph $G=(V,E)$  and a positive integer $k$. \\
\textbf{Output:} & $\lambda_k(G)$, or a correct report that $G$ is not $\lambda_k$-connected.\\
\end{tabular}
\end{boxedminipage}

\vspace{.4cm}

 The latter problem can be seen as a generalization of computing a \textsc{Minimum Cut} in a graph, which is polynomial-time solvable~\cite{StWa97}. In Section~\ref{sec:statement} we prove that it is {\sf NP}-hard, even restricted to $\lambda_k$-connected graphs. In Section~\ref{sec:parameterized} we study the parameterized complexity of the \textsc{REC} problem. More precisely, given a connected graph $G$ and two integers $k$ and $\ell$, we consider the problem of determining whether $\lambda_k(G) \leq \ell$. Existing results concerning good edge separations imply that the problem is {\sf FPT} when parameterized by $k$ and $\ell$. We prove that it is {\sf W}[1]-hard when parameterized by $k$, and that it is {\sf FPT} when parameterized by $\ell$ but unlikely to admit polynomial kernels. Moreover, we prove that the  \textsc{EREC} problem  
 is {\sf FPT} when parameterized by $k$. Finally, in Section~\ref{sec:bounded-degree} we also consider the maximum degree $\Delta$ of the input graph as a parameter, and we prove that the  \textsc{EREC} problem 
  remains {\sf NP}-complete in graphs with $\Delta \leq 5$, and that the \textsc{REC} problem is {\sf FPT} when parameterized by $k$ and $\Delta$. Note that this implies, in particular, that the \textsc{REC} problem parameterized by $k$ is {\sf FPT}  in graphs of bounded degree. Table~\ref{tab:results} below summarizes the results of this article.

\vspace{-.25cm}
\begin{table}[h]
\begin{center}
\begin{tabular}{|c||c||c|c|c|c|}
\hline
Problem & Classical & \multicolumn{4}{c|}{Parameterized complexity with parameter}\\

\cline{3-6}
 & complexity & $k + \ell$ & $k$ &  $\ell$ & $k + \Delta$  \\
 \hline  \hline

 Is $G$    &    {\sf NP}c, even     &     &   {\sf FPT}   &      &  {\sf FPT}               \\

 $\lambda_k$-connected ?        &   if $\Delta \leq 5$     &   $\star$  &  (Thm~\ref{thm:existential-FPT})     &     $\star$    &   (Thm~\ref{thm:FPTdegree})        \\
                                                            &   (Thm~\ref{thm:NPhard-bounded-degree})     &   &   &   &                 \\

 \hline
                                             &     {\sf NP}h, even if $G$  &   {\sf FPT}  &   {\sf W}[1]-hard  &  {\sf FPT} (Thm~\ref{thm:FPT-by-l})     &   {\sf FPT}               \\
$\lambda_k(G) \leq \ell$ ?         &    is  $\lambda_k$-connected   &    (Thm~\ref{thm:param_k_l},  &   (Thm~\ref{thm:W[1]-hard})   & No poly  kernels   &   (Thm~\ref{thm:FPTdegree})         \\
                                               &    (Thm~\ref{thm:guaranteed})   &   by~\cite{CCH+12})  &     &    (Thm~\ref{thm:nopolykernel})  &                 \\
\hline
\end{tabular}
\end{center}
\caption{\label{tab:results} Summary of our results, where $\Delta$ denotes the maximum degree of the input graph $G$, and {\sf NP}c  (resp. {\sf NP}h) stands for {\sf NP}-complete (resp. {\sf NP}-hard). The symbol `$\star$' denotes that the problem is not defined for that parameter.\vspace{-.65cm}}
\end{table}

\textbf{Further research}. Some open questions are determining the existence of polynomial kernels for the \textsc{REC} problem with parameters $k + \ell$ or $k + \Delta$, speeding-up the {\sf FPT} algorithm of Theorem~\ref{thm:FPT-by-l}
 (which is quite inefficient), improving the bound on the maximum degree in Theorem~ \ref{thm:NPhard-bounded-degree}, and studying the (parameterized) complexity of the \textsc{REC} problem in planar graphs and other sparse graph classes.

\vspace{.15cm}

\textbf{Notation}. We use standard graph-theoretic notation; see for instance~\cite{Diestel05}. For a graph $G$, let $\Delta(G)$ denote its maximum degree, and for a vertex $v$, its degree in $G$ is denoted by $d_G(v)$. If $S \subseteq V(G)$, we define $G-S = G[V(G) \setminus S]$, and if $S \subseteq E(G)$, we define $G-S = (V(G), E(G) \setminus S)$. Unless stated otherwise, throughout the article $n$ denotes the number of vertices of the input graph of the problem under consideration. We will always assume that the input graphs are connected.

\section{Preliminary results}
\label{sec:statement}

Clearly, any connected graph $G$ is $\lambda_1$-connected, and $\lambda_1(G)$ can be computed in polynomial time by a \textsc{Minimum Cut} algorithm (cf.~\cite{StWa97}). However, for $k \geq 2$, there exist infinitely many connected graphs which are not $\lambda_k$-connected, such as the graphs containing a cut
vertex $u$ such that every component of $G-u$ has at most $k-1$ vertices (these graphs are called \emph{flowers} in the literature\cite{BUV02}, and correspond exactly to stars when $k=2$). Moreover, the  \textsc{EREC} problem 
is hard. Indeed, given a graph $G$, if $n$ is even and $k =n/2$, by~\cite[Theorem 2.2]{DyFr85} it is {\sf NP}-complete to determine whether $G$ contains two vertex-disjoint connected subgraphs of order $n/2$ each. We can summarize this discussion as follows.

%

\begin{remarkLI}\label{rem:NPhard}
The  \textsc{EREC} problem 
  is {\sf NP}-complete.
\end{remarkLI}

In Section~\ref{sec:bounded-degree} we will strengthen the above hardness result to the case where the maximum degree of the input graph is at most 5.

%
%

Note that Remark~\ref{rem:NPhard} implies that the \textsc{REC} problem problem is {\sf NP}-hard. Furthermore, even if the input graph $G$ is guaranteed to be $\lambda_k$-connected, computing $\lambda_k(G)$ remains hard, as shown by the following theorem.

%
%
%

\begin{theorem}\label{thm:guaranteed}
The \textsc{REC} problem is {\sf NP}-hard restricted to $\lambda_k$-connected graphs.
\end{theorem}
\begin{proof} We prove it for $n$ even and $k = n/2$. The reduction is from the \textsc{Minimum Bisection} problem\footnote{Given a graph $G$ with even number of vertices, the \textsc{Minimum Bisection} problem consists in partitioning $V(G)$ into two equally-sized parts minimizing the number of edges with one endpoint in each part.} restricted to connected 3-regular graphs, which is known to be {\sf NP}-hard~\cite{BeKa01}. Given a 3-regular connected graph $G$ with even number of vertices as instance of \textsc{Minimum Bisection}, we construct from it an instance $G'$ of \textsc{REC} by adding two non-adjacent universal vertices $v_1$ and $v_2$. Note that $G'$ is $\lambda_{n/2}$-connected, since any bipartition of $V(G')$ containing $v_1$ and $v_2$ in different parts induces two connected subgraphs.

We claim that $v_1$ and $v_2$ should necessarily belong to different connected subgraphs in any optimal solution in $G'$. Indeed, let $(V_1,V_2)$ be a bipartition of $V(G)$ such that $|[V_1,V_2]|  = \lambda_{n/2}(G')$, and assume for contradiction that $v_1,v_2 \in V_1$. Since $G$ is connected, there is a vertex $u \in V_2$ with at least one neighbor in $V_1 \setminus \{v_1,v_2\}$. Let $V_1' := V_1 \cup \{u\} \setminus \{v_2\}$ and $V_2' := V_2 \cup \{v_2\} \setminus \{u\}$, and note that both $G[V_1']$ and $G[V_2']$ are connected. Since $u$ has at least one neighbor in $V_1 \setminus \{v_1,v_2\}$, $G$ is 3-regular, and $v_1$ and $v_2$ are non-adjacent and adjacent to all other vertices of $G'$, it can be checked that $|[V_1',V_2']| \leq |[V_1,V_2]| -1 = \lambda_{n/2}(G') -1$, contradicting the definition of $\lambda_{n/2}(G')$.

Therefore, solving the \textsc{REC} problem in $G'$ corresponds exactly to solving the \textsc{Minimum Bisection} problem in $G$, concluding the proof.\end{proof}


\vspace{-.35cm}

\section{A parameterized analysis}
\label{sec:parameterized}


The {\sf NP}-hardness  results of the previous section naturally lead to considering parameterized versions of the problem. In this section we consider the following three distinct parameterizations.

\vspace{.4cm}
\begin{boxedminipage}{.95\textwidth}
\textsc{Parameterized Restricted Edge-connectivity ($\mathbf{p}$-REC)}\vspace{.1cm}

\begin{tabular}{ r l }
\textbf{~~~~Instance:} & A connected graph $G=(V,E)$ and two integers $k$ and $\ell$. \\
\textbf{Parameter 1:} & The integers $k$ and $\ell$.\\
\textbf{Parameter 2:} & The integer $k$.\\
\textbf{Parameter 3:} & The integer $\ell$.\\
\textbf{Question:} & $\lambda_k(G) \leq \ell $ ?\\
\end{tabular}
\end{boxedminipage}\vspace{.4cm}


As mentioned in the introduction, determining whether $\lambda_k(G) \leq \ell$ corresponds exactly to determining whether $G$ admits a $(k-1,\ell)$-good edge separation. This latter problem has been recently shown to be solvable in time $2^{O(\min\{k,\ell\} \log(k+\ell))} \cdot n^3 \log n$ by Chitnis \emph{et al}.~\cite[Lemma II.2]{CCH+12}.


\begin{theorem}[Chitnis \emph{et al}.~\cite{CCH+12}]
\label{thm:param_k_l}
The $\mathbf{p}$-\textsc{REC} problem is {\sf FPT} when parameterized by both $k$ and $\ell$.
\end{theorem}

We would like to note that any improvement on the running time of the algorithm behind Theorem~\ref{thm:param_k_l} would answer an open question raised in~\cite{CKP13,CFJ+14}, and would have direct consequences and improve the algorithms described in~\cite{CCH+12,CLP+14,KOP+15}.

As pointed out in~\cite{EsHa88,Hol13}, the $\mathbf{p}$-\textsc{REC} problem can be solved in time $O^*(n^{2k})$. Roughly, the idea is to guess two sets of $k$ vertices inducing a connected subgraph, contract them into two vertices $s$ and $t$, and then call a polynomial-time \textsc{Minimum Cut} algorithm between $s$ and $t$ (cf.~\cite{StWa97}). In other words, it is in \textsc{XP} when parameterized by $k$. The following theorem shows that this is essentially the best algorithm we can hope for when the parameter is only $k$. Indeed, since the blow-up of the parameter in the reduction is linear, the fact that $k$-\textsc{Clique} cannot be solved in time $f(k) \cdot n^{o(k)}$ unless an unlikely collapse occurs in parameterized complexity theory~\cite{ChenHKX06} implies that the $\mathbf{p}$-\textsc{REC} problem cannot be solved in time $f(k) \cdot n^{o(k)}$ either.

\begin{theorem}\label{thm:W[1]-hard}
The $\mathbf{p}$-\textsc{REC} problem is \emph{{\sf W}[1]}-hard when parameterized by $k$.
\end{theorem}
\begin{proof} We reduce from $k$-\textsc{Clique}, which is known to be {\sf W}[1]-hard~\cite{DF99}. The parameterized reduction is the same
as the one given by Downey \emph{et al}. in~\cite[Theorem 2]{DEFPR03} to show the
{\sf W}[1]-hardness of the \textsc{Cutting $k$ Vertices from a Graph}
problem, only the analysis changes.

Let $G=(V,E)$ be an $n$-vertex
graph for which we wish to determine whether it has a $k$-clique.
We construct a graph $G'$ as follows:
\begin{itemize}
\item [(1)] We start with a clique $C$ of size $n^3$ 
and $n$ \emph{representative} vertices corresponding bijectively
with the vertices of $G$.
\item [(2)] Every representative vertex $v$ is connected to
$n^2-d_G(v)$ arbitrary vertices of $C$. 
\item [(3)] If $uv\in E(G)$ then $uv\in E(G')$.
\end{itemize}

\begin{figure}[htb]
\begin{center}\vspace{-.6cm}
 \includegraphics[width=.4\textwidth]{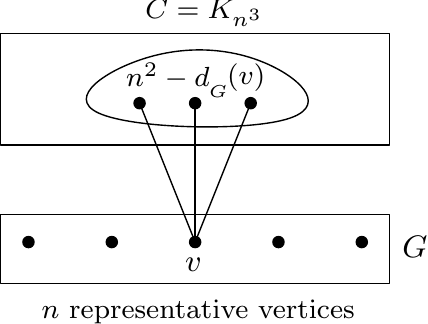}
 \caption{\label{fig:rep}Illustration of the graph $G'$ in the proof of Theorem~\ref{thm:W[1]-hard}.\vspace{-.45cm}}
\end{center}
\end{figure}

\vspace{-.45cm}

See Fig.~\ref{fig:rep} for an illustration of $G'$. Consider $\ell=kn^2-2{k\choose 2}$ and take $k\le n/2$. We claim that
$G$ has a $k$-clique if and only if $G'$ is a \textsc{Yes}-instance of
$\mathbf{p}$-\textsc{REC}.

Suppose first that $K\subseteq V(G)$ is a $k$-clique in $G$.
Obviously, $K$ is connected in $G'$ and has $k$ vertices. On the other hand, $G'-K$ is also connected with at least $n^3-|K|>k$
vertices. Finally, it is straightforward to check that
$|[K,V(G') \setminus K]|=kn^2-2{k\choose 2}=\ell$.

In the other direction, suppose $G'$ has a $k$-restricted edge-cut with at
most $\ell$ edges, i.e., there exists $K\subseteq V(G')$ such that $G[K]$
and $G'-K$
 are connected, $|K|\ge k$,  $|V(G') \setminus K|\ge k$,  and $|[K,V(G') \setminus K]|\le \ell$. Two
 cases need to be distinguished.

\emph{Case 1.} $C\cap K=\emptyset$. Then every vertex of $K$
must be a representative vertex. Hence $|[K,V(G') \setminus K]|=|K|n^2-2|E(G'[K])|$
since every representative vertex has degree $n^2$. As by hypothesis
$|[K,V(G') \setminus K]|\le \ell= kn^2-2{k\choose 2}$,  it follows that $|E(G'[K])| = {k\choose 2}$, hence $K$ must be a $k$-clique.

\emph{Case 2.} $C\cap K\neq\emptyset$. Note that for every bipartition $(C_1,C_2)$ of $C$ we have that $|[C_1,C_2]|\ge n^3-1>\ell$. Now suppose $C\cap (V(G') \setminus K)\neq\emptyset$ and  consider the bipartition $C_1=C\cap
K$ and $C_2=C\cap (V(G') \setminus K)$ of $C$. Then
$|[K,V(G') \setminus K]|\ge|[C_1,C_2]|\ge n^3-1>\ell$, a contradiction. Therefore, we have that
$C\cap (V(G') \setminus K)=\emptyset$. The proof concludes by applying Case 1 to $V(G') \setminus K$ instead of $K$. \end{proof}


In contrast to Theorem~\ref{thm:W[1]-hard} above, we now prove that  the  \textsc{EREC} problem  
 (which is {\sf NP}-complete by Remark~\ref{rem:NPhard}) is {\sf FPT} when parameterized by $k$. The proof uses the technique of \emph{splitters} introduced by Naor \emph{et al}.~\cite{NSS95}, which has also been recently used for designing parameterized algorithms in~\cite{KOP+15,CCH+12,CLP+14}. Our main tool is the following lemma.

\begin{lemma}[Chitnis \emph{et al}.~\cite{CCH+12}]
\label{lem:splitters}
There exists an algorithm that given a set $U$ of size
$n$ and  two integers
$a,b\in[0,n]$, outputs in time $2^{O(\min\{a,b\}\cdot\log(a+b))}\cdot n \log n$ a set ${\cal F}\subseteq 2^{U}$
with $|{\cal F}|=2^{O(\min\{a,b\}\cdot\log(a+b))}\cdot \log n$ such that for every two sets $A,B\subseteq U$, where $A\cap B=\emptyset$, $|A|\leq a$, and $|B|\leq b$, there exists a set $S\in{\cal F}$ with $A\subseteq S$ and $B\cap S=\emptyset$.
\end{lemma}

\begin{theorem}\label{thm:existential-FPT}
The  \textsc{EREC} problem 
 is {\sf FPT} when parameterized by $k$. More precisely, it can be solved by an algorithm running in time $2^{O(k \log k)}\cdot n^2  \log n$.
\end{theorem}
\begin{proof}
We use the easy property that $G$ is $\lambda_k$-connected if and only if $G$ contains two vertex-disjoint trees $T_1$ and $T_2$ such that $|V(T_1)| \geq k$ and $|V(T_2)| \geq k$. In order to apply Lemma~\ref{lem:splitters}, we take $U= V(G)$ and $a=b=k$, obtaining in time $k^{O(k)}\cdot n  \log n$ the desired family $\mathcal{F}$ of subsets of vertices of $G$. Now, if such trees $T_1$ and $T_2$ exist, then necessarily there exists a set $S \in \mathcal{F}$ such that
$V(T_1)\subseteq S$ and $V(T_2)\cap S=\emptyset$. Therefore, in order to determine whether $G$ is $\lambda_k$-connected or not, it suffices to check, for each set $S \in \mathcal{F}$, whether both $G[S]$ and $G - S$ contain a connected component with at least $k$ vertices. (Note that for each such set $S \in \mathcal{F}$, this can be done in linear time.) Indeed, if such a set $S$ exists, then clearly $G$ is $\lambda_k$-connected. Otherwise, by the property of the family $\mathcal{F}$, $G$ does not contain two disjoint trees $T_1$ and $T_2$ of size $k$ each, and therefore $G$ is not $\lambda_k$-connected. \end{proof}


%

Concerning the parameterized complexity of the $\mathbf{p}$-\textsc{REC} problem, in view of Theorems~\ref{thm:param_k_l} and~\ref{thm:W[1]-hard},  it just remains to settle the case when the parameter is $\ell$ only. The following theorem provides an answer to this question. We will need the following result, which is a reformulation of~\cite[Corollary 1]{BevernFSS13}.

\begin{lemma}[van Bevern \emph{et al}.~\cite{BevernFSS13}]
\label{lem:connected-bisection} Given a graph
$G$ on $n$ vertices and an integer $\ell$, determining whether $\lambda_{\PartIntInf{n/2}}(G) \leq \ell$ can be solved in time $f(\ell) \cdot n^{11}$ for some explicit function $f$ depending only on $\ell$.
\end{lemma}

\begin{theorem}\label{thm:FPT-by-l}
The $\mathbf{p}$-\textsc{REC} problem is \emph{{\sf FPT}} when parameterized by $\ell$.
\end{theorem}
\begin{proof}
 If $k \leq \ell$, we solve the problem using the {\sf FPT} algorithm given by Theorem~\ref{thm:param_k_l} with parameter $k + \ell \leq 2 \ell$. Otherwise, in the case where $k > \ell$, we proceed to Turing-reduce the problem to the particular case where $k = \PartIntInf{n/2}$, which is {\sf FPT}with parameter $\ell$ by Lemma~\ref{lem:connected-bisection} above, as follows. For each vertex $v$ of the $n$-vertex input graph $G$, and for each integer $p$ with $0 \leq p \leq n - 2k$, let $G_v^p$ be the graph obtained from $G$ by adding a clique $K_p$ on $p$ vertices and all the edges between vertex $v$ and the vertices in $K_p$.

 \begin{claim}\label{claim:reduction}
 It holds that $\lambda_k(G) \leq \ell$ if and only if there exist a vertex $v \in V(G)$ and an integer $p$ with $0 \leq p \leq n - 2k$ such that $\lambda_{\PartIntInf{\frac{n+p}{2}}}(G_v^p) \leq \ell$.
 \end{claim}\vspace{-.2cm}
 \begin{proofof}
 Assume first that  $\lambda_k(G) \leq \ell$. Let $(X,\overline{X})$ be a partition of $V(G)$ achieving $\lambda_k(G)$, where we assume without loss of generality that $X \geq \overline{X}$, let $p = |X| - |\overline{X}|$, and let $v$ be any vertex in $\overline{X}$. Then $p \leq (n - k) - k = n - 2k$ and by construction it holds that $\lambda_{\PartIntInf{\frac{n+p}{2}}}(G_v^p) \leq \lambda_k(G) \leq \ell$.

Conversely, suppose that there exist a vertex $v \in V(G)$ and an integer $p$ with $0 \leq p \leq n - 2k$ such that $\lambda_{\PartIntInf{\frac{n+p}{2}}}(G_v^p) \leq \ell$. Let $(X,\overline{X})$ be a partition of $V(G_v^p)$ achieving $\lambda_{\PartIntInf{\frac{n+p}{2}}}(G_v^p)$, and assume without loss of generality that $v \in \overline{X}$.  We claim that the clique $K_p$ is entirely contained in $\overline{X}$. Indeed, suppose for contradiction that $K_p \cap X \neq \emptyset$, and let $K = K_p \cap X$. Since $G_v^p[X]$ is connected and $v \in \overline{X}$, necessarily $X = K$. We distinguish two cases. If $p < k$, then  $|X| = |K| \leq  p -1 \leq k - 2 \leq n/2 -2 <  \PartIntInf{\frac{n+p}{2}}$, contradicting the definition of $\lambda_{\PartIntInf{\frac{n+p}{2}}}(G_v^p)$.  Otherwise, if $p \geq k$, then we use that the number of edges in $[X,\overline{X}]$ is at least the minimum cut of the clique of size $p+1$ induced by $K_p \cup \{v\}$, which is equal to $p$. That is, $|[X,\overline{X}]| \geq p \geq k > \ell$, contradicting the hypothesis that $(X,\overline{X})$ is a partition of $V(G_v^p)$ achieving $\lambda_{\PartIntInf{\frac{n+p}{2}}}(G_v^p) \leq \ell$.

We now claim that the partition of $V(G)$ given by $(X, \overline{X} \setminus K_p)$ defines a $k$-restricted edge-cut of $G$ with at most $\ell$ edges, concluding the proof. First note that both $G[X]$ and $G[\overline{X} \setminus K_p]$ are connected, since both $G_v^p[X]$ and $G_v^p[\overline{X}]$ are connected by hypothesis, and the removal of $K_p$ from $G_v^p[\overline{X}]$ clearly preserves connectivity. On the other hand, we have that $|X| \geq \PartIntInf{\frac{n+p}{2}} \geq \PartIntInf{\frac{n}{2}} \geq k$ and $|\overline{X} \setminus K_p| = |\overline{X}| - p \geq \PartIntInf{\frac{n+p}{2}} - p \geq \frac{n-p-1}{2} \geq \frac{n-(n-2k)-1}{2} = k - \frac{1}{2}$, and since both $|\overline{X} \setminus K_p|$ and $k$ are integers, the latter inequality implies that $|\overline{X} \setminus K_p| \geq k$. Finally, since $K_p \subseteq \overline{X}$, it holds that $|[X, \overline{X} \setminus K_p]| = |[X, \overline{X}]| \leq \ell $.\end{proofof}

Note that the above claim yields the theorem, as it implies that the problem of deciding whether $\lambda_k(G) \leq \ell$ can be solved by invoking $O(n^2)$ times the {\sf FPT} algorithm given by Lemma~\ref{lem:connected-bisection}.
\end{proof}

To complement Theorem~\ref{thm:FPT-by-l}, in the next theorem we prove that the $\mathbf{p}$-\textsc{REC} problem does not admit polynomial kernels when parameterized by $\ell$, unless  $\text{coNP} \subseteq \text{NP}/\text{poly}$.

\begin{theorem}\label{thm:nopolykernel}
Unless $\text{\emph{coNP}} \subseteq \text{\emph{NP}}/\text{\emph{poly}}$, the $\mathbf{p}$-\textsc{REC} problem does not admit polynomial kernels when parameterized by $\ell$.
\end{theorem}
\begin{proof} The proof is strongly inspired by the one given by van Bevern \emph{et al}.~\cite[Theorem 3]{BevernFSS13} to prove that the \textsc{Minimum Bisection} problem does not admit polynomial kernels, which in turn resembles the proof given by Garey \emph{et al}.~\cite{GJS76} to prove the {\sf NP}-hardness of \textsc{Minimum Bisection}. The main difference with respect to the proof given in~\cite{BevernFSS13} is that we need to make the appropriate modifications to guarantee that both parts left out by the edge-cut are {\sl connected}, which is not an issue in the \textsc{Minimum Bisection} problem.

We will first rule out the existence of polynomial kernels for the generalization of $\mathbf{p}$-\textsc{REC} where the edges have non-negative integer weights, and the objective is to decide whether the input graph can be partitioned into two connected subgraphs with at least $k$ vertices each by removing a set of edges whose total weight does not exceed $\ell$. We call this problem \textsc{Edge-Weighted $\mathbf{p}$-REC}. Then it will just remain to get rid of the edge weights. This is done at the end of the proof of the theorem. 


As shown by Bodlaender \emph{et al}.~\cite{BJK14}, in order to prove that \textsc{Edge-Weighted $\mathbf{p}$-REC} does not admit polynomial kernels when parameterized by $\ell$ (assuming that $\text{coNP} \subseteq \text{NP}/\text{poly}$), it is sufficient to define a \emph{cross composition} from an {\sf NP}-hard problem to \textsc{Edge-Weighted $\mathbf{p}$-REC}. In our case, the {\sf NP}-hard problem is \textsc{Maximum Cut} (see~\cite{GareyJ79comp}), which is defined as follows. Given a graph $G=(V,E)$ and an integer $p$, one has to decide whether $V$ can be partitioned into two sets $A$ and $B$ such that there are at least $p$ edges with an endpoint in $A$ and an endpoint in $B$.

A cross composition from \textsc{Maximum Cut} to \textsc{Edge-Weighted $\mathbf{p}$-REC} parameterized by $\ell$ consists of a polynomial-time algorithm that, given $t$ instances $(G_1,p_1), \ldots, (G_t,p_t)$ of \textsc{Maximum Cut}, constructs an instance $(G^*,k,\ell)$ of \textsc{Edge-Weighted $\mathbf{p}$-REC} such that $(G^*,k,\ell)$ is a \textsc{Yes}-instance if and only if one of the $t$ instances of \textsc{Maximum Cut} is a \textsc{Yes}-instance, and such that $\ell$ is polynomially bounded as a function of $\max_{1 \leq i \leq t}|V(G_i)|$. Similarly to~\cite{BevernFSS13}, we may safely assume that for each $1 \leq i \leq t$ we have $|V(G_i)|=:n$ and $p_i=:p$ (by the arguments given in~\cite{BJK14}),  that $1 \leq p \leq n^2$ (as if $p=0$ all instances are \textsc{Yes}-instances, and if $p > n^2$ all instances are \textsc{No}-instances), and that $t$ is odd (otherwise, we can add a \textsc{No}-instance consisting of an edgeless graph on $n$ vertices).

Given $(G_1,p), \ldots, (G_t,p)$, we create $G^*$ as follows; see Fig.~\ref{fig:nopoly} for an illustration. Let $w_1 := 5n^2$ and $w_2 := 5$. For each graph $G_i = (V_i, E_i)$ add to $G^*$ the vertices in $V_i$ and a clique $V_i'$ with $|V_i|=n$ vertices whose edges have weight $w_1$. Add an edge of weight $w_1$ between each vertex in $V_i$ and each vertex in $V_i'$. For each pair of vertices $u,v \in V_i$, add the edge $\{u,v\}$ to $G^*$ with weight $w_1 - w_2$ if $\{u,v\} \in E_i$, and with weight $w_1$ otherwise. Let $s_i^1,s_i^2$ be two arbitrary distinct vertices in $V_i'$, which we call \emph{link} vertices. For $1 \leq i \leq t-1$, add two edges with weight 1 between $s_{i}^1$ and $s_{i+1}^1$ and between $s_{i}^2$ and $s_{i+1}^2$, and two edges with weight 1 between $s_{t}^1$ and $s_{1}^1$ and between $s_{t}^2$ and $s_{1}^2$. This completes the construction of $G^*$. These $2t$ edges among distinct $V_i'$'s are called \emph{chain} edges (cf. the thicker edges in Fig.~\ref{fig:nopoly}). Finally, we set $k := |V(G^*)|/2$ and $\ell := w_1  n^2 - w_2 p + 4$. Note that $k$ is {\sl not} polynomially bounded in terms of $n$, but this is not a problem since the parameter we consider is $\ell$, which is bounded by $5n^4$. This construction can be clearly performed in polynomial time in $t \cdot n$. We claim that $(G^*,k,\ell)$ is a \textsc{Yes}-instance of \textsc{Edge-Weighted $\mathbf{p}$-REC} if and only if there exists $i \in \{1,\ldots,t\}$ such that $(G_i, p)$ is a  \textsc{Yes}-instance of \textsc{Maximum Cut}.

\begin{figure}[htb]
\begin{center}
 \includegraphics[width=1.05\textwidth]{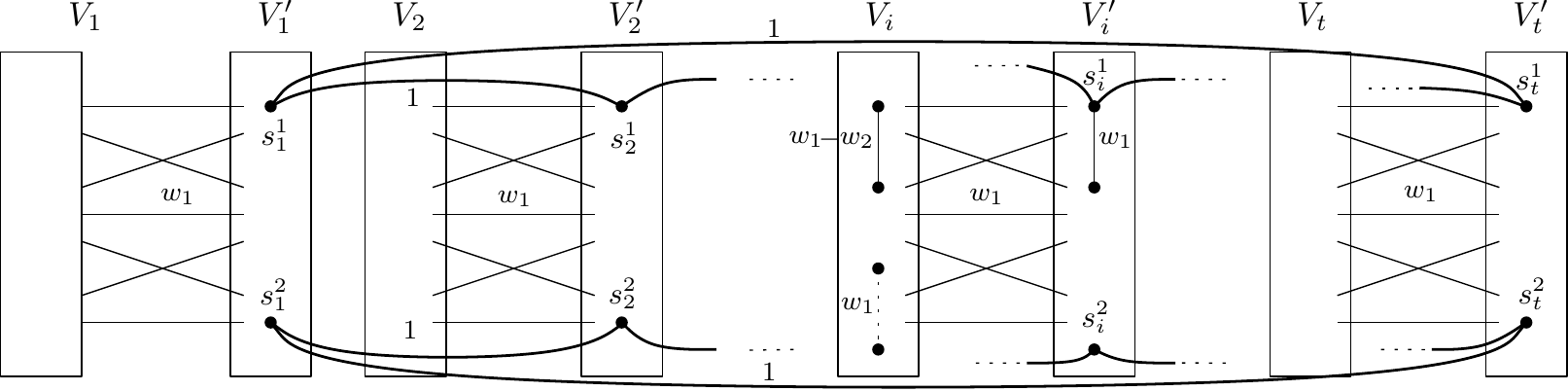}
 \caption{\label{fig:nopoly} Illustration of the graph $G^*$ in the proof of Theorem~\ref{thm:nopolykernel}.\vspace{-.55cm}}
\end{center}
\end{figure}

Assume first that there exists $i \in \{1,\ldots,t\}$ such that $(G_i, p)$ is a  \textsc{Yes}-instance of \textsc{Maximum Cut}. Assume without loss of generality that $i =1$, and let $V_1 = A \uplus B$ such that there are at least $p$ edges in $E_1$ between $A$ and $B$. We proceed to partition $V(G^*)$ into two equally-sized sets $A'$ and $B'$ such that both $G^*[A']$ and $G^*[B']$ are connected, and such that the total weight of the edges in $G^*$ with one endpoint in $A'$ and one endpoint in $B'$ is at most $\ell$. The set $A'$ contains $V_1 \cap A$, any set of $|B|$ vertices in $V_1'$ containing exactly one of $s_1^1$ and $s_1^2$ (this is possible since $1 \leq |B| \leq n-1$), and $\bigcup_{i=2}^{\PartIntSup{t/2}}V_i \cup V_i'$. Then $B' = V(G^*) \setminus A'$. Since $t$ is odd, $|A'| = |B'|$. Let us now see that $G^*[A']$ is connected.  As $1 \leq |A| \leq n-1$, the set $V_1 \cap A'$ is connected to $V_1' \cap A'$, which is connected to $V_2'$ since $A'$ contains exactly one of the link vertices $s_1^1$ and $s_1^2$. The graph $G^*[\bigcup_{i=2}^{\PartIntSup{t/2}}V_i \cup V_i']$ is clearly connected because of the chain edges, which implies that $G^*[A']$ is indeed connected. The proof for the connectivity of $G^*[B']$ is similar, using that $V_1 \cap B'$ is connected to $V_1' \cap B'$ since $1 \leq |B| \leq n-1$, which is in turn connected to $V_t'$ since $B'$ contains exactly one of the link vertices $s_1^1$ and $s_1^2$. Finally, let us show that the total weight of the edges between $A'$ and $B'$ is at most $\ell$. Note first that two chain edges incident to $V_1'$ and two chain edges incident to $V_{\PartIntSup{t/2}}'$ belong to the cut defined by $A'$ and $B'$, and no other chain edge belongs to the cut. Beside the chain edges, only edges in the graph $G^*[V_1 \cup V_1']$ are cut. Note that $G^*[V_1 \cup V_1']$ is a clique on $2n$ vertices and each of $(V_1 \cup V_1') \cap A'$ and $(V_1 \cup V_1') \cap B'$ contains $n$ vertices. Since $|[A,B]| \ge p$, 
at least $p$ of the edges of weight $w_1 - w_2$ belong to the cut. Therefore, the total weight of the cut is at most $w_1n^2 - w_2p +4 = \ell$.

Conversely, assume that for all $i \in \{1,\ldots,t\}$, $(G_i,p)$ is a \textsc{No}-instance of \textsc{Maximum Cut}, and we want to prove that  $(G^*,k,\ell)$ is a \textsc{No}-instance of \textsc{Edge-Weighted $\mathbf{p}$-REC}. Let $A \uplus B$ be a partition of $V(G^*)$ such that $|A|=|B|$ and both $G^*[A]$ and $G^*[B]$ are connected, and such that the weight of the edges between $A$ and $B$ is minimized among all such partitions. For $1 \leq i \leq t$, we let $a_i := |(V_i \cup V_i') \cap A|$. Since for $1 \leq i \leq t$, $(G_i,p)$ is a \textsc{No}-instance of \textsc{Maximum Cut}, any bipartition of $V_i$ cuts at most $p-1$ edges. Therefore, the total weight of the edges between $(V_i \cup V_i') \cap A$ and $(V_i \cup V_i') \cap B$ is at least $w_1 a_i (2n - a_i) - (p-1)w_2$.

Since $t$ is odd, necessarily at least one of the graphs $G_i$ is cut by $A \uplus B$. Assume first that exactly one graph $G_i$ is cut by $A \uplus B$. Since $|A|=|B|$, we have that $a_i = n$, so the value of the cut is at least $w_1n^2 - (p-1)w_2 = w_1n^2 - pw_2 + w_2 > w_1n^2 - pw_2 + 4 = \ell$, and thus $(G^*,k,\ell)$ is a \textsc{No}-instance of \textsc{Edge-Weighted $\mathbf{p}$-REC}.

We claim that there is always exactly one graph $G_i$ cut by $A \uplus B$. Assume for contradiction that it is not the case, that is, that there are two strictly positive values $a_i,a_j$ for some $i \neq j$. By symmetry between $A$ and $B$, we may assume that $a_i + a_j \leq 2n$. The total weight of the edges cut in $G^*[V_i \cup V_i']$ and $G^*[V_j \cup V_j']$ is at least

\[\begin{array}{lc}
w_1a_i(2n - a_i) - (p-1)w_2 + w_1a_j(2n - a_j) - (p-1)w_2 &=\\
2nw_1(a_i+a_j) - w_1(a_i^2 + a_j^2) -2w_2(p-1).&
\end{array}\]

Now we construct another solution $A' \uplus B'$ of \textsc{Edge-Weighted $\mathbf{p}$-REC} in $G^*$ where the $a_i + a_j$ vertices are cut in only one of $V_i \cup V_i'$ and $V_j \cup V_j'$, say $V_i \cup V_i'$ (note that this is possible since $a_i + a_j \leq 2n$). In order to do so, as far as there exists such a pair $a_i, a_j$, we proceed as follows.  If $a_i + a_j = 2n$, then $V_i \cup V_i'$ is entirely contained in $A'$ or $B'$. Otherwise, if $a_i + a_j < 2n$, then we choose $(V_i \cup V_i') \cap A'$ such that it contains exactly one of $s_i^1$ and $s_i^2$. At the end of this procedure, exactly one graph $G_i$ is cut. Finally, we arrange the other $G_{\ell}$'s, with $\ell \neq i$, consecutively into $A'$ and $B'$. That is, we put the vertices of $\bigcup_{\ell=i+1}^{i + \PartIntInf{t/2}}V_{\ell} \cup V_{\ell'}$ into $A'$, and the vertices of $\bigcup_{\ell=i + \PartIntSup{t/2}}^{i -1}V_{\ell} \cup V_{\ell'}$ into $B'$, where the indices are counted cyclically from $1$ to $t$. The connectivity of  both $G^*[A']$ and $G^*[B']$ is guaranteed by the chain edges and the choice of the selector vertices $s_i^1$ and $s_i^2$.


Let $i, j$ be two indices for which the above procedure has been applied. Taking into account that each $V_i'$ has  four incident chain edges, the total weight of the edges cut in $G^*[V_i \cup V_i']$ and $G^*[V_j \cup V_j']$ by the new solution is at most
\[\begin{array}{l}
w_1(a_i + a_j) (2n - a_i - a_j) + 8 \ =\\
2nw_1(a_i+a_j) - w_1(a_i^2 + a_j^2) -2w_1a_ia_j + 8.
\end{array}\]

That is, the weight of the cut defined by $A \uplus B$ minus the weight of the cut defined by $A' \uplus B'$ is at least
\[\begin{array}{rcl}
-2w_2(p-1) + 2w_1a_ia_j - 8 & = &\\
2(w_1a_ia_j - w_2(p-1) - 4) & \geq & \\
2 (w_1 - w_2(n^2-1) -4 ) & > &0,
\end{array}\]


where we have used that $a_i,a_j \geq 1$, $p \leq n^2$, $w_1 = 5n^2$, and $w_2 = 5$. In other words, $A' \uplus B'$ defines a cut of strictly smaller weight, contradicting the definition of $A \uplus B$.

\vspace{.2cm}
To conclude the proof of the theorem, it just remains to deal with the edge weights. As in~\cite{BevernFSS13}, we show how to convert the instance $(G^*,k,\ell)$ of \textsc{Edge-Weighted $\mathbf{p}$-REC} that we just constructed into an equivalent instance of \textsc{$\mathbf{p}$-REC} such that the resulting parameter remains polynomial in $n$. Given $(G^*,k,\ell)$, we define $(\hat{G},k,\ell)$ as the instance of \textsc{$\mathbf{p}$-REC}, where $\hat{G}$ is an unweighted graph obtained from  $G^*$ as follows. We replace each vertex $v$ of $G^*$ with a clique $C_v$ of size $w_1 + \ell + 1$, and for each edge $\{u,v\}$ of $G^*$ with weight $w$, we add $w$ pairwise disjoint edges between the cliques $C_u$ and $C_v$. Since no cut of size at most $\ell$ in $\hat{G}$ can separate a clique $C_v$ introduced for a vertex $v$, it follows that $(G^*,k,\ell)$ is a \textsc{Yes}-instance of \textsc{Edge-Weighted $\mathbf{p}$-REC}  if and only if $(\hat{G},k,\ell)$ is a \textsc{Yes}-instance of \textsc{$\mathbf{p}$-REC}. Finally, it is clear that the desired cut size $\ell$ is still polynomial in $n$.\end{proof}



\section{Considering the maximum degree as a parameter}
\label{sec:bounded-degree}

Towards understanding the parameterized complexity of the \textsc{REC} problem, one may wonder whether considering the maximum degree of the input graph as an extra parameter turns the problem easier (this is a classical approach in parameterized complexity, see for instance~\cite{Mar06,MaPh14}). We first prove that, from a classical complexity point of view, bounding the degree of the input graph does not turn the problem easier. Before stating the hardness result, we need the define the \textsc{3-Dimensional Matching} problem, \textsc{3DM} for short.

%


%

An instance of \textsc{3DM} consists of a set  $W=R\cup B\cup Y$, where $R, B, Y$ are disjoint sets with $|R|=|B|=|Y|=m$, and a set of triples $T\subseteq R\times B\times Y$. The question is whether there exists a matching $M\subseteq T$ covering $W$, i.e., $|M|=m$ and each element of $W=R\cup B\cup Y $ occurs in exactly one
triple of $M$.

An instance of \textsc{3DM} can be represented by a bipartite graph $G_I=(W\cup T, E_I)$,  where $E_I=\bigcup\limits_{t=(r,b,y)\in T}\{\{r,t\},\{b,t\},\{y,t\}\}$; see Fig.~\ref{3DM}.

\begin{figure}[htb]
\begin{center}
 \includegraphics[width=1.05\textwidth]{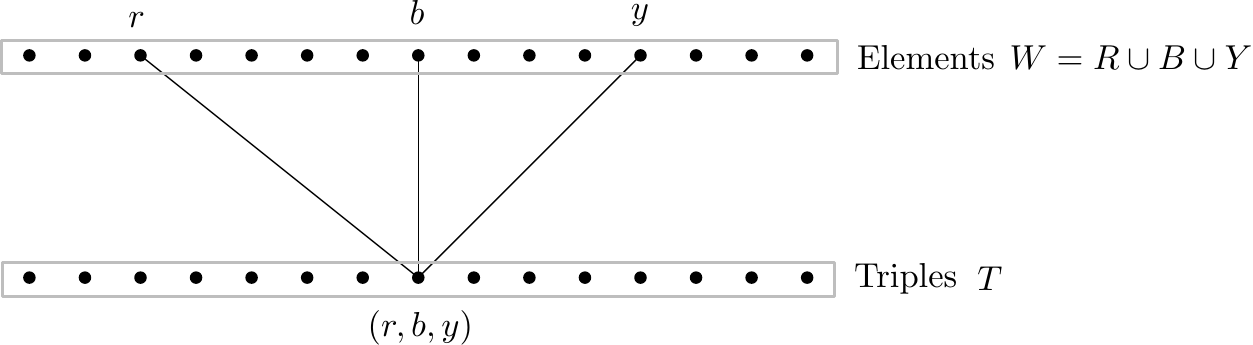}
 \caption{\label{3DM}Representation of an instance of \textsc{3DM}.}
\end{center}
\end{figure}


It is known that 3DM is {\sf NP}-complete  even if each element of $W$ appears in
$2$ or $3$ triples only~\cite{DMFA86,DyFr85}.  In \cite[Theorem 2.2]{DyFr85} it is proved that  partitioning a graph $G$ into two connected subgraphs of equal size 
is {\sf NP}-hard, using a reduction from 3DM. It is worth noting that the graph constructed in the {\sf NP}-hardness reduction contains only two vertices of degree greater than five. In Theorem \ref{thm:NPhard-bounded-degree} we appropriately modify the reduction of~\cite[Theorem 2.2]{DyFr85} so that the constructed graph has maximum degree at most 5.


\begin{theorem}\label{thm:NPhard-bounded-degree}
The  \textsc{EREC} problem 
  is {\sf NP}-complete  even if the maximum degree of the input graph is 5.
\end{theorem}
\begin{proof} Given an instance $(W,T)$ of 3DM with $W=R\cup B\cup Y$, $|R|=|B|=|Y|=m$, and $T\subseteq R\times B\times Y$ such that each element of $W$ appears in
$2$ or $3$ triples only, we define an $n$-vertex graph $G=(V,E)$ with maximum
degree 5 as follows (see Fig.~\ref{fig:reduction} for an illustration).

\begin{figure}[htb]
\begin{center}
 \includegraphics[width=.85\textwidth]{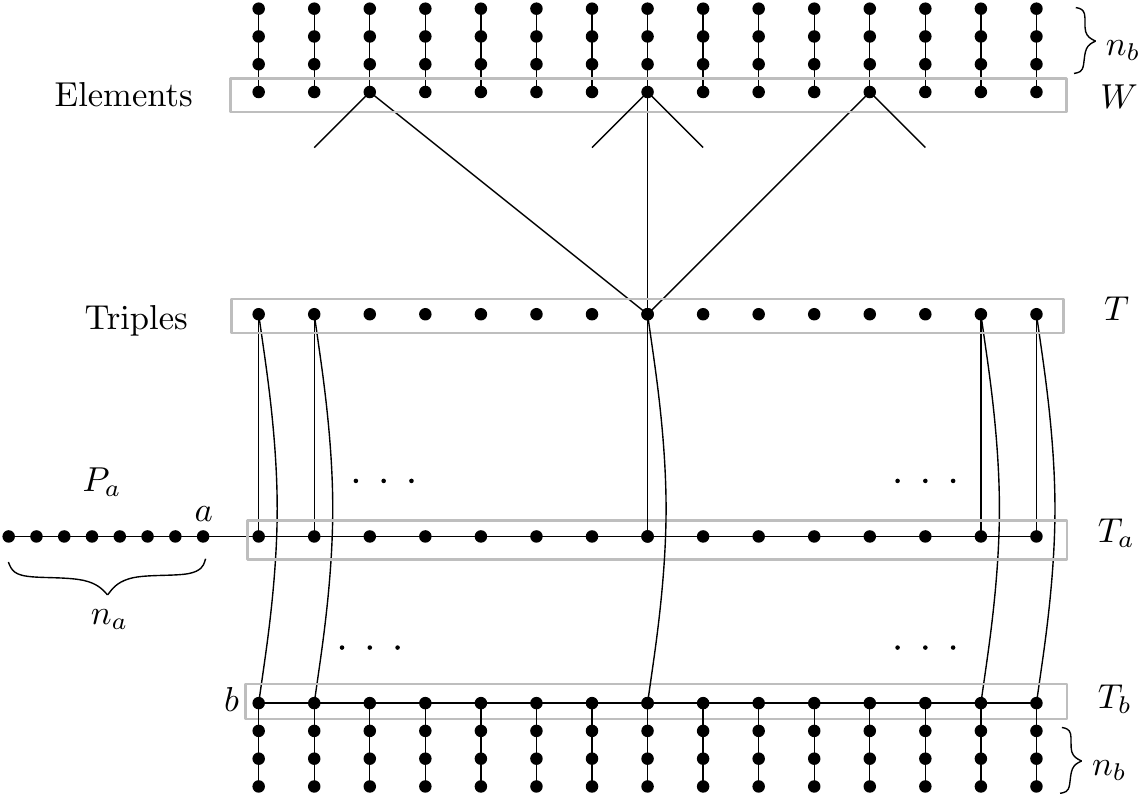}
 \caption{\label{fig:reduction}Construction of the graph $G$ in the proof of Theorem~\ref{thm:NPhard-bounded-degree}, with $\Delta(G)=5$.}
\end{center}
\end{figure}


The set of vertices of $G$ is $$V=W\cup
T\cup T_a\cup T_b\cup\mathcal{P}\cup \{a\},$$ where
  $T_a=\{t^a_1,\ldots,t^{a}_{|T|}\}$,
$T_b=\{b=t^b_1,t^b_2,\ldots,t^b_{|T|}\}$, $T=\{t_1,\ldots,t_{|T|}\}$ is the set of triples, and  $\mathcal{P}=\bigcup\limits_{\sigma\in W\cup T_b\cup\{a\}} P_{\sigma}$, where   $P_{\sigma}=\{(\sigma,t) : t=1,\ldots,n_{\sigma}\}$ with
$n_a=(3m+|T|)n_b+5m-|T|-1$, $n_b=2m^3$, and $n_{\sigma}=n_b$ for
every $\sigma\in W\cup T_b$.

The set of edges of $G$ is
$$E=E_I \cup E_{T_a} \cup E_{T_b}\cup E_{T^+}\bigcup\limits_{\sigma\in W\cup T_b\cup\{a\}}E_{\sigma},$$ where $E_{T_a}=\{\{t_i^a,t_{i+1}^a\} : 1\le i\le|T|-1\}$,
$E_{T_b}=\{\{t_i^b,t_{i+1}^b\} : 1\le i\le|T|-1\}$, $E_{T^+}=\{\{t_i,t_i^a\},\{t_i,t_i^b\} : 1\le i\le|T|\}$, and $E_{\sigma}=\{\{\sigma,(\sigma,1)\}\}\cup
\{\{(\sigma,t),(\sigma,t+1)\} : 1\le t\le n_{\sigma}-1\}\cup
\{a,t^a_1\}$ for every $\sigma\in W\cup T_b\cup\{a\}.$

Note that the maximum degree of $G$ is indeed 5 (only vertices of $T$ could get degree $5$, all other vertices have degree at most $4$). Since $n=1+3m+3|T|+n_a+(3m+|T|)n_b$, we can observe that $$n=2(n_a+1+2|T|-m).$$

Next, we show that for $k=n/2$,  $G$ is \textsc{Yes}-instance of the \textsc{REC} problem 
 if and only if
$T$ contains a matching covering $W$.

One direction is easy. Suppose first that $T$ contains a matching $M$ covering $W$. Let $S=\{a\}\cup
P_a\cup T_a\cup (T \setminus M)$. It is straightforward to check that
$|S|=n/2$ and that $G[S]$, $ G[V\setminus S]$ are both connected.

Conversely, suppose that $G$ can be partitioned into $2$ connected subgraphs
$G[S]$, $G[V\setminus S]$ with $|S|=n/2$.  We can assume that $a\in S$, and then it follows
that $P_a\subseteq S$. Now $|S \setminus (P_a\cup\{a\})|=2|T|-m<2m^3=n_b$ since $|T|\le m^3$.  As $P_{\sigma}\subseteq S$ if and only if $\sigma\in
S\cap (W\cup T_b)$, then $S\cap(W\cup T_b)=\emptyset$ since
$|S \setminus (P_a\cup\{a\})|<n_b$ and $|P_{\sigma}|=n_b$ for every $\sigma\in W\cup T_b$. Hence
$S \setminus (P_a\cup\{a\})\subseteq T\cup T_a$. Let $M=(V\setminus S)\cap T$. Then $|M|\le m$ since $|S \setminus (P_a\cup\{a\})|=2|T|-m$. Finally, as $G[V\setminus S]$ is
connected and $W\cup T_b\subseteq V\setminus S$, it follows that
$|M|\ge m$. Hence $|M|=m$ and $M$ must be a matching covering $W$. \end{proof}

In order to understand to which extent the vertices of high degree make the complexity of computing the restricted edge-connectivity of a graph hard, we also consider the maximum degree of the input graph as a parameter for the $\mathbf{p}$-\textsc{REC} problem.

%

\begin{theorem}\label{thm:FPTdegree}
The $\mathbf{p}$-\textsc{REC} problem is {\sf FPT} when parameterized by $k$ and the maximum degree $\Delta$ of the input graph.
\end{theorem}
\begin{proof}
The algorithm is based on a simple exhaustive search. We use the property that, for any graph $G$ and any two integers $k,\ell$, $\lambda_k(G) \leq \ell$ if and only if $G$ contains two vertex-disjoint trees $T_1$ and $T_2$ with $|V(T_1)| \geq k$ and $|V(T_2)| \geq k$, such that there exists an edge set $S$ in $G$ with $|S| \leq \ell$ such that in $G - S$ the trees $T_1$ and $T_2$ belong to different connected components. Hence, we just have to determine whether these trees exist in $G$ or not. For doing so, for every pair of distinct vertices $v_1$ and $v_2$ of $G$, we exhaustively consider all trees $T_1$ and $T_2$ with $k$ vertices containing $v_1$ and $v_2$, respectively. Note that the number of such trees is at most $\Delta^{2k}$. For every pair of vertex-disjoint trees $T_1$ and $T_2$, we proceed as follows. We contract  tree $T_1$ (resp. $T_2$) to a single vertex $t_1$ (resp. $t_2$), keeping edge multiplicities, and then we run in the resulting graph a polynomial-time \textsc{Minimum Cut} algorithm between $t_1$ and $t_2$ (cf.~\cite{StWa97}). If the size of the returned edge-cut is at most $\ell$, then $T_1$ and $T_2$ are the desired trees. Otherwise, we continue searching. It is clear that the overall running time of this algorithm is $O(\Delta^{2k} \cdot n^{O(1)})$.
\end{proof}

\vspace{.2cm}
\noindent \textbf{Acknowledgement}. We would like to thank the anonymous referees  for helpful remarks that improved the presentation of the manuscript. We are particularly grateful for the ideas to prove Theorem~\ref{thm:FPT-by-l}, a result that we had left as an open question in the conference version of this article.


\bibliographystyle{abbrv}
\bibliography{edge-conn-FPT}

\begin{thebibliography}{10}

\bibitem{BCFF97}
C.~Balbuena, A.~Carmona, J.~F{\`a}brega, and M.~A. Fiol.
\newblock Extraconnectivity of graphs with large minimum degree and girth.
\newblock {\em Discrete Mathematics}, 167:85--100, 1997.

\bibitem{BeKa01}
P.~Berman and M.~Karpinski.
\newblock Approximation hardness of bounded degree {MIN-CSP} and
  {MIN-BISECTION}.
\newblock {\em {Electronic Colloquium on Computational Complexity}}, 8(26),
  2001.

\bibitem{BJK14}
H.~L. Bodlaender, B.~M.~P. Jansen, and S.~Kratsch.
\newblock Kernelization lower bounds by cross-composition.
\newblock {\em SIAM Journal on Discrete Mathematics}, 28(1):277--305, 2014.

\bibitem{BUV02}
P.~Bonsma, N.~Ueffing, and L.~Volkmann.
\newblock Edge-cuts leaving components of order at least three.
\newblock {\em Discrete Mathematics}, 256(1):431--439, 2002.

\bibitem{ChenHKX06}
J.~Chen, X.~Huang, I.~A. Kanj, and G.~Xia.
\newblock Strong computational lower bounds via parameterized complexity.
\newblock {\em Journal of Computer and System Sciences}, 72(8):1346--1367,
  2006.

\bibitem{CCH+12}
R.~Chitnis, M.~Cygan, M.~Hajiaghayi, M.~Pilipczuk, and M.~Pilipczuk.
\newblock Designing {FPT} algorithms for cut problems using randomized
  contractions.
\newblock {\em {SIAM} Journal on Computing}, 45(4):1171--1229, 2016.

\bibitem{CFJ+14}
M.~Cygan, F.~Fomin, B.~M. Jansen, L.~Kowalik, D.~Lokshtanov, D.~Marx,
  M.~Pilipczuk, and M.~Pilipczuk.
\newblock {Open problems from School on Parameterized Algorithms and
  Complexity}, \texttt{http://fptschool.mimuw.edu.pl/opl.pdf}, 2014.

\bibitem{Cygan15book}
M.~Cygan, F.~V. Fomin, L.~Kowalik, D.~Lokshtanov, D.~Marx, M.~Pilipczuk,
  M.~Pilipczuk, and S.~Saurabh.
\newblock {\em Parameterized Algorithms}.
\newblock Springer, 2015.

\bibitem{CKP13}
M.~Cygan, L.~Kowalik, and M.~Pilipczuk.
\newblock {Open problems from Update Meeting on Graph Separation Problems},
  \texttt{http://worker2013.mimuw.edu.pl/slides/update-opl.pdf}, 2013.

\bibitem{CLP+14}
M.~Cygan, D.~Lokshtanov, M.~Pilipczuk, M.~Pilipczuk, and S.~Saurabh.
\newblock Minimum bisection is fixed parameter tractable.
\newblock In {\em Proc. of the 46th ACM Symposium on Theory of Computing
  (STOC)}, pages 323--332, 2014.

\bibitem{Diestel05}
R.~Diestel.
\newblock {\em Graph Theory}.
\newblock Springer-Verlag, Berlin, 3rd edition, 2005.

\bibitem{DEFPR03}
R.~G. Downey, V.~Estivill{-}Castro, M.~R. Fellows, E.~Prieto, and F.~A.
  Rosamond.
\newblock Cutting up is hard to do: the parameterized complexity of $k$-cut and
  related problems.
\newblock {\em Electronic Notes in Theoretical Computer Science}, 78:209--222,
  2003.

\bibitem{DF99}
R.~G. Downey and M.~R. Fellows.
\newblock {\em Parameterized Complexity}.
\newblock Springer-Verlag, 1999.

\bibitem{DyFr85}
M.~E. Dyer and A.~M. Frieze.
\newblock On the complexity of partitioning graphs into connected subgraphs.
\newblock {\em Discrete Aplied Mathematics}, 10:139--153, 1985.

\bibitem{DMFA86}
M.~E. Dyer and A.~M. Frieze.
\newblock Planar \textsc{3DM} is \textsc{NP}-complete.
\newblock {\em Journal of Algorithms}, 7(2):174--184, 1986.

\bibitem{EsHa88}
A.-H. Esfahanian and S.~L. Hakimi.
\newblock On computing a conditional edge-connectivity of a graph.
\newblock {\em Information Processing Letters}, 27(4):195--199, 1988.

\bibitem{FF94}
J.~F\`abrega and M.~A. Fiol.
\newblock Extraconnectivity of graphs with large girth.
\newblock {\em Discrete Mathematics}, 127:163--170, 1994.

\bibitem{FG06}
J.~Flum and M.~Grohe.
\newblock {\em Parameterized Complexity Theory}.
\newblock Springer-Verlag, 2006.

\bibitem{GareyJ79comp}
M.~R. Garey and D.~S. Johnson.
\newblock {\em Computers and intractability. A guide to the theory of
  NP-completeness}.
\newblock W. H. Freeman and Co., 1979.

\bibitem{GJS76}
M.~R. Garey, D.~S. Johnson, and L.~J. Stockmeyer.
\newblock {Some simplified NP-complete graph problems}.
\newblock {\em Theoretical Computer Science}, 1(3):237--267, 1976.

\bibitem{Hol13}
A.~Holtkamp.
\newblock {\em Connectivity in Graphs and Digraphs. Maximizing vertex-, edge-
  and arc-connectivity with an emphasis on local connectivity properties}.
\newblock PhD thesis, RWTH Aachen University, 2013.

\bibitem{HMM12}
A.~Holtkamp, D.~Meierling, and L.~P. Montejano.
\newblock $k$-restricted edge-connectivity in triangle-free graphs.
\newblock {\em Discrete Applied Mathematics}, 160(9):1345--1355, 2012.

\bibitem{KaTh11}
K.~Kawarabayashi and M.~Thorup.
\newblock {The Minimum $k$-way Cut of Bounded Size is Fixed-Parameter
  Tractable}.
\newblock In {\em Proc. of the 52nd Annual Symposium on Foundations of Computer
  Science (FOCS)}, pages 160--169, 2011.

\bibitem{KOP+15}
E.~J. Kim, C.~Paul, I.~Sau, and D.~M. Thilikos.
\newblock Parameterized algorithms for min-max multiway cut and list digraph
  homomorphism.
\newblock In {\em Proc. of the 10th International Symposium on Parameterized
  and Exact Computation (IPEC)}, volume~43 of {\em LIPIcs}, pages 78--89, 2015.

\bibitem{Mar06}
D.~Marx.
\newblock Parameterized graph separation problems.
\newblock {\em Theoretical Computer Science}, 351(3):394--406, 2006.

\bibitem{MaPh14}
D.~Marx and M.~Pilipczuk.
\newblock Everything you always wanted to know about the parameterized
  complexity of subgraph isomorphism (but were afraid to ask).
\newblock In {\em Proc. of the 31st International Symposium on Theoretical
  Aspects of Computer Science (STACS)}, volume~25 of {\em LIPIcs}, pages
  542--553, 2014.

\bibitem{NSS95}
M.~Naor, L.~J. Schulman, and A.~Srinivasan.
\newblock Splitters and near-optimal derandomization.
\newblock In {\em Proc. of the 36th Annual Symposium on Foundations of Computer
  Science (FOCS)}, pages 182--191, 1995.

\bibitem{Nie06}
R.~Niedermeier.
\newblock {\em Invitation to Fixed-Parameter Algorithms}.
\newblock Oxford University Press, 2006.

\bibitem{RRS16}
A.~Rai, M.~S. Ramanujan, and S.~Saurabh.
\newblock A parameterized algorithm for mixed-cut.
\newblock In {\em Proc. of the 12th Latin American Symposium on Theoretical
  Informatics (LATIN)}, volume 9644 of {\em LNCS}, pages 672--685, 2016.

\bibitem{StWa97}
M.~Stoer and F.~Wagner.
\newblock A simple min-cut algorithm.
\newblock {\em Journal of the ACM}, 44(4):585--591, 1997.

\bibitem{BevernFSS13}
R.~van Bevern, A.~E. Feldmann, M.~Sorge, and O.~Such{\'{y}}.
\newblock On the parameterized complexity of computing graph bisections.
\newblock In {\em Proc. of the 39th International Workshop on Graph-Theoretic
  Concepts in Computer Science (WG)}, volume 8165 of {\em LNCS}, pages 76--87,
  2013.

\bibitem{YL10}
J.~Yuan and A.~Liu.
\newblock Sufficient conditions for $\lambda_k$-optimality in triangle-free
  graphs.
\newblock {\em Discrete Mathematics}, 310:981--987, 2010.

\bibitem{ZY05}
Z.~Zhang and J.~Yuan.
\newblock A proof of an inequality concerning $k$-restricted edge-connectivity.
\newblock {\em Discrete Mathematics}, 304(1-3):128--134, 2005.

\end{thebibliography}

\end{document}